\documentclass[letterpaper,11pt]{article}

\usepackage{setspace}
\usepackage{subcaption}
\usepackage{enumerate}
\usepackage{amsmath}
\usepackage{amsfonts}
\usepackage{graphicx}
\usepackage{tcolorbox}
\usepackage{amssymb}
\usepackage{multirow}
\usepackage{geometry}
\usepackage{soul}
\usepackage{colortbl}
\usepackage{wrapfig}
\usepackage{algorithm}
\usepackage{algorithmicx}
\usepackage{algpseudocode}
\usepackage{amsthm}
\newif\ifanon   \anonfalse
\newif\ifdraft  \drafttrue
\ifdraft
  \usepackage{todonotes}
\else
  \usepackage[disable]{todonotes}
\fi
\usepackage{mathtools} 
\usepackage{mathrsfs}

\usepackage[pdftex, plainpages = false, pdfpagelabels, 
                 bookmarks=false,
                 bookmarksopen = true,
                 bookmarksnumbered = true,
                 breaklinks = true,
                 linktocpage,
                 pagebackref,
                 colorlinks = true,  % was true
                 linkcolor = blue,
                 urlcolor  = blue,
                 citecolor = red,
                 anchorcolor = green,
                 hyperindex = true,
                 hyperfigures
                 ]{hyperref}

\usepackage{thmtools} 
\usepackage{thm-restate}

\newcommand{\MVD}{\Problem{Metric Violation Distance} }

\newcommand{\dist}{{\sf dist}}
\newcommand{\Problem}[1]{\textsc{#1}}

\algtext*{EndWhile}% Remove "end while" text
\algtext*{EndIf}% Remove "end if" text
\algtext*{EndFor}% Remove "end for" text
\algtext*{EndProcedure}
\algtext*{EndFunction}

\newtheorem{lemma}{Lemma}

\numberwithin{lemma}{section}

\newtheorem{theorem*}[lemma]{Theorem*}
\newtheorem{corollary}[lemma]{Corollary}
\newtheorem{definition}[lemma]{Definition}
\newtheorem{observation}[lemma]{Observation}

\newtheorem{claim}[lemma]{Claim}
\newtheorem{proposition}[lemma]{Proposition}

\usepackage{wrapfig}
\usepackage{bm}

\definecolor{problemrule}{RGB}{120,160,200}
\newtcolorbox{problembox}[1]{%
  colback=white, colframe=problemrule, boxrule=0.6pt,
  arc=2mm, outer arc=2mm,
  left=8pt, right=8pt, top=7pt, bottom=7pt,
  before upper={\textsc{#1}\par\smallskip}%
}

\title{How to Fix a Broken Metric:\\ A Linear Kernel, Tight Bounds, and Tree Metrics}
\ifanon
  \author{}
\else
  \author{%
    Sayan Bandyapadhyay\thanks{Portland State University. \texttt{sayanb@pdx.edu}}
    \and
    William Lochet\thanks{CNRS, LIRMM. \texttt{william.lochet@gmail.com}}
    \and
    Daniel Lokshtanov\thanks{UCSB. \texttt{daniello@ucsb.edu}}
    \and
    Saket Saurabh\thanks{Institute of Mathematical Sciences. \texttt{saket@imsc.res.in}}
    \and
    Chinmay Sonar\thanks{\texttt{chinmaysonar96@gmail.com}}
    \and
    Vaishali Surianarayanan\thanks{UC Santa Cruz. \texttt{vaishalisurianarayanan@gmail.com}}
    \and
    Jie Xue\thanks{New York University Shanghai. \texttt{jiexue@nyu.edu}}
  }
\fi
\date{}
\begin{document}
 
\pagenumbering{gobble}

\maketitle
\begin{abstract}
Given a complete graph whose edge weights represent dissimilarities, the \MVD
problem asks whether at most \(k\) of the weights can be changed so that the
resulting weights form a metric. Motivated by metric repair for noisy data, this problem has recently been studied primarly from the approximation algorithm community, where the best current algorithm is a \(\mathcal{O}(\log n)\)-approximation in polynomial time by
Cohen-Addad, Fan, Lee and de Mesmay [SIAM J. Comput., 2025]. 

In the context of parameterized complexity, Fan, Gilbert, Raichel, Sonthalia and Van
Buskirk proved the problem fixed-parameter tractable, in \(k^{\mathcal{O}(k)}n^{\mathcal{O}(1)}\)
time [SWAT 2020]. Fomin, Golovach and More [IPEC 2026] showed that, in the case of ultrametic, the problem has a $O(k^2)$ kernel as well as a single-exponential running time. They also ask the following two questions about general metric: can the running time be made single-exponential and does the problem admit
a polynomial kernel. They also ask what is the parameterized complexity of the tree-metric
analogue?  We answer all three questions.

We give a \(2^{\mathcal{O}(k)}n^{\mathcal{O}(1)}\)-time algorithm, and this is best possible:
unless ETH fails there is no \(2^{o(k)}n^{\mathcal{O}(1)}\)-time algorithm, even when every
input distance lies in \(\{1,2,3\}\).  We also give a kernel with at most \(6k\)
vertices.  The kernel is \emph{solution-lifting}, so it may be placed in front of
any approximation algorithm; used before the \(\mathcal{O}(\log n)\)-approximation of
Cohen-Addad, Fan, Lee and de Mesmay it yields an approximation factor of
\(\mathcal{O}(\log \mathrm{OPT})\), at no asymptotic cost in running time.  
 
Both results extend to an interval generalization in
which every edge \(e\) carries an observed value \(M_e\) and an admissible range
\([A_e,B_e]\) within which it may be reassigned; there the kernel has \(7k\)
vertices.

Finally, we show that \textsc{Tree Metric Violation Distance} is fixed-parameter
tractable, and solvable in \(k^{\mathcal{O}(k)}n^{\mathcal{O}(1)}\) time.
\end{abstract}

\newpage
\pagenumbering{arabic}
\section{Introduction}

Many algorithms in data analysis use distances in a metric space to represent dissimilarities between data points. This assumption is useful both conceptually and algorithmically: metric structure is central in clustering, metric embeddings, and many approximation algorithms. In practice, however, distances obtained from measurements, similarities, or incomplete data are often noisy, and may fail even the most basic metric constraint, namely the triangle inequality.

This motivates the problem of repairing a distance matrix. Classical metric repair, often formulated as metric nearness, asks for a nearby metric, typically with respect to an \(\ell_p\)-type objective \cite{veldt2018projection,tang2022efficient,brickell2008metric}. Such repairs may change a large fraction of the input distances. The \textsc{Metric Violation Distance} problem instead asks for a sparse repair: given pairwise distances \(\delta:E(G) \rightarrow \mathbb{R}_{\geq0} \) on a complete graph, modify as few edge lengths as possible so that the resulting distances form a metric \cite{gilbert2017sparse,fan2018metric}.  Throughout, we work with the parameterized decision version: the input is \((G,\delta,k)\), and the question is whether some \(S\subseteq E(G)\) with \(|S|\le k\) admits a reassignment of the lengths of its edges after which \(\delta\) is a metric.  The parameter is \(k\), the number of distances one is allowed to touch, and we call such an \(S\) a \emph{feasible support}.

The problem has mainly been studied from the viewpoint of approximation algorithms. Fan, Raichel, and Van Buskirk initiated this line, proving hardness results and giving an \(\mathcal{O}(\mathrm{OPT}^{1/3})\)-approximation \cite{fan2018metric}. More recently, Cohen-Addad, Fan, Lee, and de Mesmay gave an \(\mathcal{O}(\log n)\)-approximation for \textsc{Metric Violation Distance} \cite{cohenaddad2022fitting}. Whether the problem admits a constant-factor approximation remains open.

The same fitting approach has also been considered for more specific metrics, most notably ultrametrics and tree metrics~\cite{cohenaddad2022fitting,Charikar_ultra,kipouridis2023tree}. These metrics are central in applications where distances are expected to come
from an underlying hierarchy: ultrametrics model hierarchical clusterings
\cite{ailon2011fitting}, while tree metrics are the basic distance model in
phylogenetics and numerical taxonomy
\cite{sneath1962numerical,cavalli1967phylogenetic,farris1972estimating}. A constant-factor approximation is known for these two variants \cite{cohenaddad2022fitting,Charikar_ultra,kipouridis2023tree}, but to our knowledge, nothing is known regarding their parameterized complexity. 

In this paper we study the parameterized complexity of distance repair, parameterized by the number \(k\) of distances that may be modified.  Fan, Raichel and Van Buskirk \cite{fan2018metric} asked whether \MVD on complete graphs is fixed-parameter tractable in the solution size.  This was answered by Fan, Gilbert, Raichel, Sonthalia and Van Buskirk \cite{fan_et_al:LIPIcs.SWAT.2020.25}, who gave a bounded search-tree algorithm for the more general class of \(\varsigma\)-chordal graphs, those in which every chordless cycle has length at most \(\varsigma\).  Complete graphs are \(3\)-chordal, so their algorithm solves \MVD in \(k^{\mathcal{O}(k)}n^{\mathcal{O}(1)}\) time.  Recently, Fomin Golovach and More \cite{FominGolovachMore2026} raised three questions about this state of affairs:
\begin{enumerate}[(Q1)]
    \item Can the \(k^{\mathcal{O}(k)}\) dependence be improved to single-exponential?
    \item Does \MVD admit a polynomial kernel?
    \item What is the parameterized complexity of the tree-metric analogue?
\end{enumerate}
We answer all three.

Questions (Q1) and (Q2) are natural, and (Q2) has a concrete obstruction behind it: Fan et al.\ explicitly observe \cite[Sec.~5]{fan_et_al:LIPIcs.SWAT.2020.25} that the edges of the broken chordless cycles, which drive their search tree, do \emph{not} form a kernel of size bounded by the parameter.  Their example is \(K_n\) with one edge of weight \(n+1\) and all others of weight \(1\): the optimum has size \(1\), while \(2n-3\) edges survive.  A kernel must therefore come from somewhere else.

Our results hold for a natural interval generalization of the problem that we call \textsc{Interval Metric Violation}. For every edge \(e\), the input consists of an observed value \(M_e\) and an allowed interval \([A_e,B_e]\). If \(e\) is not modified, its final length must remain \(M_e\); if \(e\) is modified, its final length may be chosen anywhere in \([A_e,B_e]\). This captures the common situation where the observed distances come with uncertainty or confidence intervals. The standard \textsc{Metric Violation Distance} problem is recovered by choosing sufficiently permissive intervals for modified edges.

\subsection{Our Results and Techniques}

Our first result answers (Q1): a single-exponential algorithm, improving the \(k^{\mathcal{O}(k)}n^{\mathcal{O}(1)}\) bound of \cite{fan_et_al:LIPIcs.SWAT.2020.25}.

\begin{restatable}{theorem}{fpttheorem}\label{thm:FPT}
    \textsc{Interval Metric Violation} can be solved in \(2^{\mathcal{O}(k)} \cdot n^{\mathcal{O}(1)}\) time.
\end{restatable}

Both this and the kernel below rest on a characterization of the sets \(S\) that can serve as supports.  Every edge \(e\) has a range of admissible final lengths, namely \([A_e,B_e]\) if \(e\in S\) and the single value \(M_e\) otherwise; write \(a_S(e)\) and \(b_S(e)\) for its endpoints.  Call a cycle \emph{bad} for \(S\) if one of its edges \(h\) satisfies \(a_S(h)>\sum_{f\ne h}b_S(f)\), and call \(S\) \emph{good} if no cycle is bad for it.  Goodness is testable in polynomial time by a shortest-path computation with lengths \(b_S\), and Lemma \ref{lem:good_feasible} shows it to be equivalent to feasibility.

Section \ref{section:FPT} then proceeds by random separation.  Color the edges red and blue uniformly at random, and let \(F_\chi\) consist of the red edges whose deletion from the red set destroys goodness.  If an inclusion-minimal solution \(O\) is colored red, and the edges of the bad cycles certifying the minimality of each \(e\in O\) are colored blue, then \(F_\chi=O\) exactly. To obtain the desired running time, those certifying cycles must not be large in aggregate: writing \(Q\) for their edges outside \(O\), the coloring succeeds with probability \(2^{-|O|-|Q|}\).  Choosing the cycles first to be shortest, and then, among those, so as to minimise the number of edges far from the heavy edge that are not already accounted for, we obtain \(|Q|\le 10k\), hence success probability \(2^{-11k}\).  Universal sets derandomize the algorithm.

Our second result answers (Q2), and does better than a polynomial kernel: the number of vertices is linear in \(k\).

\begin{restatable}{theorem}{kerneltheorem}\label{thm:linear_kernel}
    \textsc{Interval Metric Violation} admits a kernel with at most \(7k\) vertices, computable in \(\mathcal{O}(n^3)\) time.  For the standard \MVD problem the bound improves to \(6k\) vertices.
\end{restatable}

Section \ref{Section:kernel} starts from the triangles that already violate the triangle inequality.  Take a maximal family of edge-disjoint such triangles; every feasible support must contain an edge of each, so more than \(k\) of them is a NO-instance, and otherwise their at most \(3k\) edges and \(3k\) vertices form a mandatory core.  What remains is to bound the rest of the graph: a vertex outside these triangles may still form an initially bad triangle with one of their edges, and there may be \(\Theta(n)\) such \emph{witnesses}.  Keeping \(k+1\) witnesses per edge is a representative-set argument, and gives an \(\mathcal{O}(k^2)\)-vertex kernel.  To reach a linear bound we keep instead only what a matching needs: form the bipartite graph between outside vertices and the core edges they witness, take a maximum matching, and delete every unmatched outside vertex.

The rule is safe by an exchange argument along alternating paths, and this yields more than equivalence of the two instances: any solution \(S\) of the kernel lifts, in \(\mathcal{O}(n^2)\) time, to a solution \(S^+\) of the original instance with \(|S^+|\le|S|\).  The kernel is therefore \emph{solution-lifting}, and can be placed in front of any approximation algorithm.  Doing so before the \(\mathcal{O}(\log n)\)-approximation of \cite{cohenaddad2022fitting} gives, for the standard problem, an approximation factor of \(\mathcal{O}(\log \mathrm{OPT})\). This improvement is not merely cosmetic: in \cite{cohenaddad2022fitting} they exhibit instances with optimum \(1\) on which their algorithm has expected cost \(\Omega(\log n)\). Moreover, the kernel runs in \(O(n^3)\) time, matching their algorithm, so the preprocessing does not increase the asymptotic running time. 

Our third result shows that the answer to (Q1) is optimal: the dependence on \(k\) in Theorem \ref{thm:FPT} cannot be improved further.

\begin{restatable}{theorem}{ethlowerbound}\label{thm:complete-eth-lower}
Unless ETH fails, \MVD on complete graphs has no algorithm with running time \(2^{o(k)}n^{\mathcal{O}(1)}\), even when all edge weights belong to \(\{1,2,3\}\).
\end{restatable}

Section \ref{section:eth} obtains this from a parameter-preserving reduction from \Problem{Vertex Cover}, which adds a single apex vertex \(r\) to the input graph and gives weight \(1\) to every edge at \(r\), weight \(3\) to every edge of the graph, and weight \(2\) to every non-edge.  The bad triangles of the construction are then exactly the triangles \(ruv\) with \(uv\) an edge, so hitting them is covering.

The last two sections move from metrics to the structured distance functions of numerical taxonomy.  A distance function \(d\) on a finite set \(X\) is a \emph{tree metric} if there is a tree \(T\) with nonnegative edge lengths and \(X\subseteq V(T)\) such that \(d(x,y)\) is the length of the unique \(x,y\)-path of \(T\); equivalently, \(d\) is a metric satisfying the \emph{four-point condition}, that for all distinct \(a,b,c,e\) the maximum of
\[
    d(ab)+d(ce),\qquad d(ac)+d(be),\qquad d(ae)+d(bc)
\]
is attained at least twice.  It is an \emph{ultrametric} if moreover \(d(xz)\le\max\{d(xy),d(yz)\}\) for all \(x,y,z\), that is, if in every triangle the maximum length is attained at least twice; equivalently, if \(T\) may be taken rooted with all points of \(X\) at the same distance from the root.  Every ultrametric is a tree metric and every tree metric is a metric, so the three classes are nested.  \Problem{Ultrametric Violation Distance} and \textsc{Tree Metric Violation Distance} ask exactly the question of \MVD with ``metric'' replaced by the corresponding class.

Section~\ref{sec:uvd} observes that a standard order-preserving blow-up of the
input distances transfers Theorems~\ref{thm:FPT} and~\ref{thm:linear_kernel} to
\Problem{Ultrametric Violation Distance} with no loss
(Corollary~\ref{cor:uvd}). Finally, Section \ref{sec:tree-metric} answers (Q3).

\begin{restatable}{theorem}{treefpttheorem}\label{thm:tree-metric-fpt}
    \textsc{Tree Metric Violation Distance} can be solved in \(k^{\mathcal{O}(k)} n^{\mathcal{O}(1)}\) time.
\end{restatable}

\section{Preliminaries}

For a graph \(G\), an edge-length function \(d:E(G)\to \mathbb{R}_{\geq 0}\), and a path \(P\), the \emph{length} of \(P\) is \(d(P)=\sum_{e\in E(P)}d(e)\). We say that \(d\) is a \emph{metric on \(G\)} if every edge is no longer than any alternative path between its endpoints: for every edge \(uv\in E(G)\) and every \(u,v\)-path \(P\) in \(G-uv\), \(d(uv)\le d(P)\). On complete graphs this is exactly the usual triangle inequality.

We now restate the main problem studied in the paper. For the rest of the paper, \(G\) is the complete graph on \(n\) vertices.  

\begin{problembox}{Interval Metric Violation}
\textbf{Input:} A complete graph \(G\), values \(0 \le A_e\le M_e\le B_e\) for every
edge \(e\), and an integer \(k\).\\
\textbf{Parameter:} \(k\).\\
\textbf{Question:} Is there a set \(S\subseteq E(G)\) with \(|S|\le k\) and
final edge lengths \(d_e\) such that \(d_e=M_e\) for \(e\notin S\),
\(d_e\in[A_e,B_e]\) for \(e\in S\), and \(d\) is a metric?
\end{problembox}

When \(A_e=0\) and \(B_e\) is large enough for every edge, this corresponds to the standard \MVD problem, which is usually given by a weight function \(\delta:E(G)\to \mathbb{R}_{\geq 0}\).  Concretely, it suffices to take \(B_e = n\cdot\max_{f\in E(G)}\delta(f)\) for every edge. One useful feature of the standard version is that the problem reduces to hitting cycles that 
violate the metric inequality. More specifically, if \(\delta\) is a length function from \(E(G)\) to \(\mathbb{R}_{\geq 0}\), then 
a cycle \(C\) in \(G\) is said to be {\em bad} if there exists an edge \(e\) in \(C\) such that \(\delta(e)>\sum_{e'\in C\setminus e} \delta(e')\).
In that case the edge \(e\) is said to be the {\em heavy} edge of the bad cycle. 
A solution \(S\) to the standard problem is then exactly a set of edges hitting every bad cycle \cite{fan_et_al:LIPIcs.SWAT.2020.25}. Assuming \(G-S\) is connected, the final weight of an edge of the solution may be taken to be the shortest weighted distance between its endpoints in \(G-S\).
One way to interpret this is that, in the standard problem, the final weight of an edge in \(S\) can be pushed as high as possible without making that edge heavy. In the interval version, this may exceed \(B_e\). We therefore push only up to the upper endpoint \(B_e\). One consequence is that bad cycles of \(G-S\) depend not only on the edges left in \(G-S\), but also on the weights allowed for the edges of \(S\). The following definitions formalize this. 

\begin{definition}
\label{def:interval-support}
A \emph{support} is a set \(S\subseteq E(G)\) of edges that are allowed to be
changed.  For a fixed support \(S\), the admissible interval of an edge \(e\) is
\[
I_S(e)=
\begin{cases}
[A_e,B_e], & e\in S,\\
[M_e,M_e], & e\notin S.
\end{cases}
\]
Let \(a_S(e)\) and \(b_S(e)\) denote the lower and upper endpoints of
\(I_S(e)\).

A support \(S\) is \emph{feasible} if there is a metric \(d\) on \(G\) with
\(d_e\in I_S(e)\) for every edge \(e\); equivalently, if the lengths of the edges
of \(S\) can be reassigned within their intervals so that the result is a metric.
A \emph{solution} is a feasible support of size at most \(k\), so the question is
whether a solution exists.
\end{definition}

We can then define bad cycles relative to a support.  

\begin{definition}

We say that a cycle \(C\) is \emph{bad} for a support \(S\) if it has an edge \(h\) such that
\[
        a_S(h)>
        \sum_{f\in E(C)\setminus\{h\}} b_S(f).
\]
We call \(h\) the heavy edge of the bad cycle.  A support is \emph{good} if it has
no bad cycle.
\end{definition}

\begin{observation}\label{obs:support-monotone}
Let \(R\) and \(R'\) be two supports, and let \(C\) be a cycle. If \(R\cap E(C)\subseteq R'\cap E(C)\) and \(C\) is bad for \(R'\), then \(C\) is bad for \(R\). In particular, if \(R\subseteq R'\) and \(R\) is good, then \(R'\) is good.

Moreover, if \(R\cap E(C)=R'\cap E(C)\), then \(C\) is bad for \(R\) if and only if it is bad for \(R'\).
\end{observation}

\begin{proof}
For every edge \(e\in E(C)\), the containment \(R\cap E(C)\subseteq R'\cap E(C)\) implies \(a_R(e)\ge a_{R'}(e)\) and \(b_R(e)\le b_{R'}(e)\), so any inequality showing that \(C\) is bad for \(R'\) also shows that it is bad for \(R\). If \(R\cap E(C)=R'\cap E(C)\), then the values \(a\) and \(b\) used in the definition of a bad cycle are the same for both supports.
\end{proof}

We let \(\dist_{b_S}(x,y)\) denote the shortest \((x,y)\) distance in the graph \(G\) when every edge \(e\) has length \(b_S(e)\). 
We say that a set \(S\) is \emph{feasible} if there exist final edge lengths \(d_e\) such that \(d_e=M_e\) for \(e\notin S\),
\(d_e\in[A_e,B_e]\) for \(e\in S\), and \(d\) is a metric. Note that if \(S\) is feasible, then so is \(S \cup \{e\}\) with the same weight function. 
We are now ready to show the following lemma, which is the interval analogue of the classical characterization described above.

\begin{lemma}\label{lem:good_feasible}
    A set \(S\) is feasible if and only if it is good. Moreover, the weight function \(d_{(xy)} :=  \dist_{b_S}(x,y)\) 
    is a metric such that  \(d_e=M_e\) for \(e\notin S\) and  \(d_e\in[A_e,B_e]\) for \(e\in S\). 
\end{lemma}

\begin{proof}
    Suppose first that \(S\) is feasible with \(d\) the associated metric. Note that by 
    definition this means that \(a_S(e)\le d(e)\le b_S(e)\) for every edge. A bad cycle \(C\) with 
    heavy edge \(h\) would satisfy \(a_S(h) > \sum_{f\in E(C)\setminus\{h\}} b_S(f)\). Since \(a_S(h) \leq d(h)\) 
    and \(\sum_{f\in E(C)\setminus\{h\}} b_S(f) \geq \sum_{f\in E(C)\setminus\{h\}} d(f)\), we immediately reach a contradiction. 

    Assume now that \(S\) is good, i.e., has no bad cycle, and consider the edge length defined as \(d_{(xy)} :=  \dist_{b_S}(x,y)\). This length is indeed a metric 
    since it is the shortest path metric on the complete graph \(G\) with weight \(b_S(e)\) on every edge \(e\). Therefore, to conclude the proof we only have to 
    verify the conditions that \(d_e=M_e\) for \(e\notin S\) and  \(d_e\in[A_e,B_e]\) for \(e\in S\).

    If \(xy \not \in S\), we have that \(a_S(xy) = b_S(xy) = M_{xy}\). Moreover, if the edge \(xy\) is not the shortest path in the graph with weight
    \(b_S(e)\) on every edge, then we would immediately obtain a bad cycle with \(xy\) as the heavy edge, which is a contradiction. Therefore \(d_{(xy)} = M_{xy}\).

    Similarly, if \(xy \in S\), let \(P_{xy}\) be a path of shortest weight. By definition we have that \(d_{xy} = \sum_{f \in P_{xy}} b_S(f)\). Since \(xy\) is a path between \(x\) and \(y\), 
    we automatically have that  \(d_{xy} \leq B_{xy}\). Finally, if \(d_{xy} < a_S(xy)\), then \(xy\) and \(P_{xy}\) form a bad cycle for \(S\), 
    which is not possible. 
\end{proof}

\section{FPT algorithm}\label{section:FPT}

Our FPT algorithm will use random separation to find the set \(O\). For that purpose we need the notion of private bad cycle, which will be used to certify that an edge \(e\) must belong
to \(O\). More precisely, let \(O\), with \(|O| \leq k\), denote an inclusion-minimal feasible solution. Lemma \ref{lem:good_feasible} tells us that no cycle is bad for \(O\), and by minimality of \(O\),
we also have that \(O_e := O \setminus \{e\}\) must contain a bad cycle. By Observation \ref{obs:support-monotone}, this cycle must contain \(e\): otherwise \(O\) and \(O_e\) agree on all of its edges, and it would also be bad for \(O\). We say that such a cycle is a \emph{private bad cycle} for \(e\). For every \(e \in O\), let \(C_e\) denote an arbitrary private bad cycle with heavy edge \(h_e\), and let \(Q := \left( \bigcup _{e \in O} E(C_e) \right) \setminus O\). Before arguing about the size of \(Q\) which will be necessary for 
bounding the complexity of our algorithm, let us explain how a random separation can be used to recover the set \(O\). 

\begin{lemma}\label{lem:color}
    Let \(\chi\) denote a red-blue coloring of \(E(G)\) such that the set \(R_\chi\) of red edges contains \(O\) and the set \(B_\chi\) of blue 
    edges contains \(Q\). \(O\) is exactly the set of edges \(e\) such that \(R_\chi \setminus \{e\}\) is not good.  
\end{lemma}

\begin{proof}
    If \(e \in O\) and the edges of \(E(C_e) \setminus O\) are in \(B_\chi\), then \(R_\chi \setminus \{e\}\) agrees with \(O_e\) on the edges of \(C_e\). Hence, by Observation \ref{obs:support-monotone}, \(C_e\) is a bad cycle in \(R_\chi \setminus \{e\}\). 
    Conversely, for any \(e \not \in O\),  \(R_\chi \setminus \{e\}\) contains \(O\), which is a good support for the instance. By Observation \ref{obs:support-monotone}, \(R_\chi \setminus \{e\}\) is good. 
\end{proof}

Therefore, in order for a random separation to allow us to find the solution, we only need to know its value on the elements of \(O\) and \(Q\), for some choice of private bad cycles. This means that, if we can show that \(Q\) has size linear in \(k\), then we immediately obtain the desired algorithm, as a uniform red-blue coloring has probability \(1/2^{|Q \cup O|}\) of satisfying the properties of Lemma \ref{lem:color}. This is the purpose of the next lemma. 

\begin{lemma}\label{lem:size_Q}
    There exists a choice of private bad cycles \(C_e\) over \(e \in O\) such that
    \[
        Q := \left( \bigcup _{e \in O} E(C_e) \right) \setminus O
    \]
    has size at most \(10k\).
\end{lemma}

\begin{proof}
    Let us start with a choice of \(C_e\) for each \(e \in O\) which is minimal, i.e., one that uses the fewest edges among all the bad cycles of \(O_e\). 

    \begin{claim}\label{clm:chord}
            Any chord \(q\) of \(C_e\) not adjacent to \(h_e\) belongs to \(O_e\).
    \end{claim}

    \begin{proof}
    Indeed, suppose that \(q \not \in O_e\) and let \(P_2\) be the path between the endpoints of \(q\) that does not contain 
    \(h_e\). By minimality of \(C_e\), we have that \(q \cup P_2\) is not a bad cycle; in particular, this means that \(a_{O_e}(q) \leq \sum_{f \in P_2} b_{O_e}(f)\). However, since \(q \not \in O_e\), we have that \(a_{O_e}(q) = b_{O_e}(q)\)
    and we can therefore replace \(P_2\) with \(q\) and still obtain a bad cycle. This contradicts the minimality of \(C_e\). 
\end{proof}

We will now partition each cycle \(C_e\) into two subpaths \(P_{e,1}\) and \(P_{e,2}\) such that: 
\begin{itemize}
    \item \(P_{e,1}\) is a path or a cycle of length at most \(4\) containing \(h_e\)
    \item \(P_{e,2}\) is a path of length either \(0\) or at least \(2\) whose endpoints are not adjacent to \(h_e\). 
\end{itemize}
To do so, let \(P_e\) be the path \(C_e - h_e\) between the endpoints of \(h_e\). If \(P_e\) has length at most \(3\), then we set \(P_{e,1}\) as the entire cycle \(C_e\) and \(P_{e,2}\) as the empty path. Otherwise, we let \(P_{e,1}\) consist of \(h_e\) and the two adjacent edges on \(C_e\), and \(P_{e,2}\) be the rest of the cycle. 

If \(|V(G)|\le 3\), then the proof is immediate, so assume \(|V(G)|\ge 4\). For every \(xy \in O\), let \(\alpha_1(xy)\) and  \(\alpha_2(xy)\) denote the two lightest length-two \(x,y\)-paths for the edge weights \(b_O\), taken with distinct internal vertices, and let \(H := \bigcup_{xy \in O}\bigcup_{i \in [2]}E(\alpha_i(xy))\). Note that for a fixed pair \((x,y)\) these two paths are edge disjoint because the intermediate vertices are different. In particular, at most one of them contains any fixed edge; we will use this to avoid \(h_e\). 
The following claim is the main part of the proof. It shows that, after choosing the private cycles carefully, every edge of the part of \(C_e\) far from the heavy edge is either already in the global replacement set \(H\), belongs to \(O\), or lies next to \(e\). The proof is a shortcut argument: if two consecutive edges \(fg\) violate this, then the chord joining the endpoints of \(fg\) must belong to \(O\), and one of the two cheapest two-edge \(x,y\)-paths gives an equally short private bad cycle with fewer offending edges.

\begin{claim}\label{clm:far-edges}
    Each \(C_e\) can be chosen so that every edge of \(P_{e,2} \setminus (H \cup O)\) is adjacent to \(e\) on \(P_{e,2}\).  (If \(e\notin E(P_{e,2})\), this asserts that \(P_{e,2} \setminus (H \cup O)=\emptyset\).)
\end{claim}

\begin{proof} 
    Indeed, fix \(e \in O\) and choose \(C_e\), among the private bad cycles with the minimum number of edges, to minimize the number of edges of \(P_{e,2} \setminus (H \cup O)\) that are not adjacent to \(e\) on \(P_{e,2}\). Assume for contradiction that such an edge \(f\) exists. Let \(g\) be an edge adjacent to \(f\) on \(P_{e,2}\) with \(g \neq e\), which exists by the choice of \(f\), and let \(x,y\) denote the endpoints of the path \(fg\). The edge \(xy\) is a chord of \(C_e\), and it is not adjacent to \(h_e\) because \(fg\) lies on \(P_{e,2}\). By Claim \ref{clm:chord}, \(xy \in O_e \subseteq O\), so let us consider \(\alpha_1(xy)\) and \(\alpha_2(xy)\). Let \(\beta\) be \(\alpha_1(xy)\) if it does not contain \(h_e\), and \(\alpha_2(xy)\) otherwise.
 
    The path \(\beta\) does not contain \(h_e\). Moreover, since the path \(fg\) also avoids \(h_e\), the choice of \(\alpha_1(xy)\) and \(\alpha_2(xy)\) gives \(b_O(\beta) \leq b_O(f)+b_O(g)\): this is immediate if \(\beta=\alpha_1(xy)\), and if \(\beta=\alpha_2(xy)\) then \(\alpha_1(xy)\) contains \(h_e\) while \(fg\) does not, so \(fg\) is a length-\(2\) path distinct from \(\alpha_1(xy)\). Since \(f,g \neq e\), we have \(b_O(f)=b_{O_e}(f)\) and \(b_O(g)=b_{O_e}(g)\); also \(b_{O_e}(\beta) \leq b_O(\beta)\). Replacing \(fg\) by \(\beta\) therefore gives a walk \(W\) between the endpoints of \(h_e\) with \(b_{O_e}(W) < a_{O_e}(h_e)\). The walk \(W\) has the same number of edges as \(C_e-h_e\). Removing closed subwalks from \(W\) gives a simple path between the endpoints of \(h_e\) of no larger \(b_{O_e}\) weight. This simple path together with \(h_e\) gives another bad cycle of \(O_e\) whose length is at most \(|C_e|\). This cycle must contain \(e\), since otherwise \(O \) and \(O_e\) would agree on all its edges and it would be a bad cycle of \(O\) by Observation \ref{obs:support-monotone}.  
     
    Finally, if the extracted cycle is shorter than \(C_e\), this contradicts the choice of \(C_e\). On the other hand, if the extracted cycle has the same length as \(C_e\), then no closed subwalk was removed. Thus the replacement really removed \(f\) and inserted only edges of \(H\), contradicting the secondary choice of \(C_e\). Indeed, replacing \(P_{e,2}\) by the path \(W\) will still give a bad cycle with heavy edge \(h_e\) for \(O_e\), and by the properties of \(O_e\), this must be a private bad cycle of \(e\).
\end{proof}
It remains to count the edges of \(Q\). Indeed, as \(O\) is a set of size \(k\), \(H\) contains at most \(4k\) edges, and the union of the \(P_{e,1}\) has size at most \(4k\). Moreover, for each \(e \in O\), at most two edges of \(P_{e,2} \setminus (H \cup O)\) are adjacent to \(e\) on \(P_{e,2}\): an edge of a path has at most two neighbours on it, and if \(e\notin E(P_{e,2})\) there are none. Thus \(Q\) has size at most \(4k+4k+2k=10k\), which concludes the proof.
\end{proof}

We are now ready to prove Theorem \ref{thm:FPT}.

\fpttheorem*

\begin{proof}[Proof of Theorem \ref{thm:FPT}]
The algorithm picks a random coloring \(\chi\) of the edges into red or blue uniformly, and lets \(R_\chi\) be the set of red edges. If \(R_\chi\) is not good, it rejects the coloring. Otherwise, it computes
\[
    F_\chi := \{e \in E(G) \mid R_\chi \setminus \{e\} \text{ is not good}\}.
\]
If \(|F_\chi|\leq k\) and \(F_\chi\) is good, then it accepts and returns \(F_\chi\); otherwise, it rejects the coloring. Each goodness test can be done in polynomial time: a support \(T\) is not good if and only if, for some edge \(h=xy\), the shortest \(x,y\) path in \(G-h\) with edge weights \(b_T\) has length smaller than \(a_T(h)\).

The algorithm is sound by Lemma \ref{lem:good_feasible}: whenever it accepts, \(F_\chi\) is a good support of size at most \(k\), and hence is feasible. For completeness, suppose that the instance is a YES instance, and let \(O\) be an inclusion minimal feasible solution. If \(\chi\) colors the edges of \(O\) red and the edges of \(Q\) blue, then Lemma \ref{lem:color} gives \(F_\chi=O\), so the algorithm accepts. By Lemma \ref{lem:size_Q}, this event has probability at least \(1/2^{|Q| + |O|}\geq 1/2^{11k}\). Repeating the process \(2^{\mathcal{O}(k)}\) times gives constant success probability on every YES instance, and the standard universal-sets derandomization gives a deterministic \(2^{\mathcal{O}(k)}\cdot n^{\mathcal{O}(1)}\) algorithm.
\end{proof}

\section{Kernelization}\label{Section:kernel}

For a vertex set \(Y\), let \(\mathcal I[Y]\) denote the instance induced by \(Y\), with all intervals inherited from the original instance and with the same parameter \(k\). We first give a representative \(\mathcal{O}(k^2)\)-kernel as a warm-up, and then explain how to replace the representatives by a matching argument to obtain a linear kernel.

Call a triangle \emph{initially bad} if it is bad for the empty support. Let \(\mathcal T\) be a maximal family of edge-disjoint initially bad triangles. If \(|\mathcal T|>k\), then we return a fixed NO-instance, since every feasible support must contain one edge from each triangle of \(\mathcal T\). Otherwise, let
\[
        F:=\bigcup_{T\in\mathcal T}E(T),
        \qquad
        U:=V(F).
\]
Then \(|F|\le 3k\), \(|U|\le 3k\), and by maximality every initially bad triangle contains an edge of \(F\).  We call the edges of \(F\) the \emph{core edges}, the set \(U\) the \emph{core}, and a vertex of \(V(G)\setminus U\) an \emph{outside vertex}.

\begin{claim}\label{clm:free-vertex-shortening}
Let \(Y\subseteq V(G)\), let \(R\subseteq E(G[Y])\), and let \(C\) be a shortest cycle in \(G[Y]\) that is bad for \(R\). Then either \(C\) is a triangle, or every vertex of \(C\) is incident with some edge of \(R\).
\end{claim}

\begin{proof}
Suppose that \(C\) is not a triangle, and let \(z\in V(C)\) be a vertex incident with no edge of \(R\). Let \(zg\) be a chord from \(z\) to a nonconsecutive vertex of \(C\). This chord is not in \(R\), so \(a_R(zg)=b_R(zg)=M_{zg}\). Let \(P_1,P_2\) be the two \(z,g\)-paths on \(C\), and assume that a heavy edge \(h\) of \(C\) lies on \(P_1\). If \(b_R(zg)\le b_R(P_2)\), then \(P_1\cup\{zg\}\) is a shorter bad cycle with heavy edge \(h\). If \(b_R(zg)>b_R(P_2)\), then \(P_2\cup\{zg\}\) is a shorter bad cycle with heavy edge \(zg\). Both cases contradict the choice of \(C\).
\end{proof}

For a core edge \(f=ab\in F\) and a vertex \(z\notin U\), we say that \(z\) is a \emph{witness} for \(f\) if \(abz\) is initially bad. Such a witness is \emph{\(f\)-dominated} if \(abz\) is not bad for the support \(\{f\}\). An outside vertex \(z\notin U\) is \emph{non-dominated} if it is a witness for some \(f\in F\), but is not \(f\)-dominated.

\begin{claim}\label{clm:nondominated}
If there are more than \(k\) non-dominated outside vertices, then the instance is a NO-instance.
\end{claim}

\begin{proof}
Let \(z\notin U\) be non-dominated, and let \(f=ab\in F\) witness this. Then \(abz\) is initially bad and remains bad for the support \(\{f\}\). Hence every feasible support \(S\) must contain one of the two side edges \(az,bz\): if neither side edge is in \(S\), then the restriction of \(S\) to the triangle \(abz\) is either empty or \(\{f\}\), and in both cases the triangle is bad. The side edges forced by different outside vertices are distinct, because each has one endpoint equal to its outside vertex and the other endpoint in \(U\). Thus more than \(k\) non-dominated outside vertices cannot be covered by a support of size at most \(k\).
\end{proof} 

Let \(\mathcal B\) be the set of non-dominated outside vertices. By Claim \ref{clm:nondominated}, either we already returned NO or \(|\mathcal B|\le k\).  For the standard problem this set is in fact empty.

\begin{proposition}\label{prop:representative_kernel}
    \textsc{Interval Metric Violation} admits a kernel with \(\mathcal{O}(k^2)\) vertices.
\end{proposition}

\begin{proof}
For each core edge \(f\in F\), keep \(k+1\) arbitrary \(f\)-dominated witnesses, or all of them if there are fewer than \(k+1\). Let \(X\) be \(U\cup \mathcal B\) together with all kept witnesses. Since \(|F|\le 3k\) and \(|\mathcal B|\le k\), we have \(|X|=\mathcal{O}(k^2)\).

If \(O\) is a good support in the original instance, then \(O\cap E(G[X])\) is good in \(\mathcal I[X]\). Conversely, let \(S\subseteq E(G[X])\), \(|S|\le k\), be good in \(\mathcal I[X]\). Suppose that \(S\) is not good in the original instance, and take a shortest cycle \(C\) that is bad for \(S\). Since \(S\) is good on \(X\), the cycle contains a vertex \(z\notin X\). No edge incident with \(z\) belongs to \(S\), so Claim \ref{clm:free-vertex-shortening} implies that \(C\) is a triangle, say \(xyz\).

Since \(C\) is bad for \(S\) and \(z\not \in \mathcal{B} \), Observation \ref{obs:support-monotone} shows that it is initially bad. By maximality of \(\mathcal T\), it contains an edge of \(F\); since \(z\notin X\) and \(U\subseteq X\), this edge must be the edge \(xy\) opposite \(z\). Thus \(z\) is a witness for \(xy\). Moreover \(z\notin \mathcal{B}\), so \(z\) is \(xy\)-dominated. If \(xy\in S\), then the restriction of \(S\) to the triangle \(xyz\) is exactly \(\{xy\}\), so \(xyz\) is not bad, a contradiction. Hence \(xy\notin S\).

Because \(z\notin X\), not all \(xy\)-dominated witnesses were kept, so \(X\) contains \(k+1\) kept \(xy\)-dominated witnesses. Therefore, one such witness, say \(z'\), is such that neither \(xz'\in S\) nor \(yz'\in S\). Then \(xyz'\) is initially bad, and none of its three edges belongs to \(S\). Hence \(xyz'\) is bad for \(S\) inside \(X\), contradicting that \(S\) is good in \(\mathcal I[X]\). Therefore \(S\) is good globally.
\end{proof}

The linear kernel keeps the same mandatory set \(U\cup \mathcal B\), but avoids storing \(k+1\) dominated witnesses for every core edge. Instead, we keep only the vertices that are still needed by the following matching argument. Set
\[
        W:=V(G)\setminus (U\cup \mathcal B),
\]
and build the bipartite graph
\[
        H=(W,F;E_H),
\]
where \(z\in W\) is adjacent to \(f\in F\) if \(z\) is a witness for \(f\). By definition of \(\mathcal B\), every witness \(z\in W\) for a core edge \(f\) is \(f\)-dominated.
 
We now apply the following reduction rule: compute a maximum matching \( M \) in \( H \) and let \( W_M \) be the set of vertices of \( W \) saturated by \( M \). The kernel simply removes \( R := W \setminus W_M \) from the instance.

\begin{claim}\label{clm:reduction_safe}
        The instance \(\mathcal I[ U \cup \mathcal B \cup W_M]\) is equivalent to the original instance. In fact any solution to  \(\mathcal I[ U \cup \mathcal B \cup W_M]\) can be lifted, in polynomial time, to a solution of the original instance of the same size. 
\end{claim}

\begin{proof}
The forward direction is immediate, since restricting to an induced subinstance cannot create a bad cycle. For the reverse direction, let \(S\) be a good support for \(\mathcal I[ U \cup \mathcal B \cup W_M]\).

Consider the directed alternating graph obtained from \(H\) by orienting every edge of \(M\) from \(F\) to \(W\), and every edge of \(E(H)\setminus M\) from \(W\) to \(F\). Let \(A\subseteq W\) and \(P\subseteq F\) be the vertices reachable from \(R\) in this directed graph, on the two sides respectively.

We have \(P=N_H(A)\). Indeed, if \(z\in A\) and \(f\in N_H(z)\), then either \(zf\notin M\), in which case the alternating path reaching \(z\) extends to \(f\), or \(zf\in M\), in which case \(z\notin R\) and the alternating path reaching \(z\) must enter \(z\) through the matching edge \(fz\). In both cases \(f\in P\).

Moreover, by maximality of \(M\), every vertex of \(P\) is matched. Otherwise an alternating path from \(R\) to an unmatched vertex of \(F\) would be an augmenting path. The matching partner of every \(f\in P\) belongs to \(A\setminus R\), and conversely every vertex of \(A\setminus R\) is matched to a vertex of \(P\). Thus \(M\) gives a bijection between \(P\) and \(A\setminus R\). For \(f\in P\), let \(m(f)\in A\setminus R\) denote its matched vertex.

Now define
\[
        D:=\{e\in S:e\text{ is incident with a vertex of }A \setminus R\}.
\]
For every \(f=ab\in P\setminus S\), the vertex \(m(f)\) is a witness for \(f\), so the triangle \(abm(f)\) is initially bad. Since \(f\notin S\), goodness of \(S\) forces \(S\) to contain one of the side edges \(a m(f), b m(f)\), which belongs to \(D\). These side edges are distinct for different \(f\), hence
\(|D|\ge |P\setminus S|\).
Now set
\[
        S^+:=(S\setminus D)\cup P.
\]
The edges of \(P\) have both endpoints in \(U\), while every edge of \(D\) is incident with \(A \setminus R\subseteq W\). Therefore \(|S^+| \le |S|\) and  
no edge of \(S^+\) is incident with any vertex of \(A \).

It remains to show that \(S^+\) is good for the original instance. Suppose not,
and let \(C\) be a shortest cycle that is bad for \(S^+\). If \(C\) avoids
\(A\), then it also avoids \(R\), since \(R\subseteq A\). Hence \(C\) is
contained in the kept instance \(\mathcal I[U\cup \mathcal B\cup W_M]\). Moreover, no
edge of \(D\) lies on \(C\), so on the edges of \(C\) the support \(S^+\) is
obtained from \(S\) only by possibly adding edges of \(P\). By
Observation~\ref{obs:support-monotone}, \(C\) would also be bad for \(S\) in
\(\mathcal I[U\cup \mathcal B\cup W_M]\), a contradiction. Hence \(C\) contains some \(z\in A\). Since no edge of \(S^+\) is incident
with \(z\), Claim~\ref{clm:free-vertex-shortening} implies that \(C\) is a
triangle.

Since \(C\) is bad for \(S^+\), Observation \ref{obs:support-monotone} shows that \(C\) is initially bad. Thus, by maximality of \(\mathcal T\), it contains an edge \(f\in F\). As \(z\notin U\), this core edge is opposite \(z\), and \(z\) is a witness for \(f\). Since \(z\in A\subseteq W\), we have \(f\in N_H(A)\), so \(f\in P \subseteq S^+\). But every witness in \(W\) is dominated, and no side edge incident with \(z\) belongs to \(S^+\). Therefore the restriction of \(S^+\) to \(C\) is exactly \(\{f\}\), which makes \(C\) non-bad, a contradiction.

Thus \(S^+\) is a good support of size at most \(|S|\) for the original instance, and deleting \(R\) preserves equivalence. Moreover, \(S^+\) can easily be computed from \(S\). Indeed, since the maximum matching \(M\) is known, computing \(D\) and \(P\) can be done in \(\mathcal{O}(n^2) \) time. 
\end{proof}

\kerneltheorem*

\begin{proof}
If \(n\le 6k\), then the original instance is already a kernel of the desired size, in both the interval and the standard case, so assume \(k<n/6\). The NO-answers above are correct: edge-disjoint initially bad triangles force distinct support edges, and Claim \ref{clm:nondominated} handles the non-dominated vertices. Thus we may assume \(|\mathcal T|\le k\) and \(|\mathcal B|\le k\). 

Let \(M\) be a maximum matching in \(H\), and let \(W_M\) be the set of vertices of \(W\) saturated by \(M\). By Claim \ref{clm:reduction_safe}, the instance induced by
\[
        Y:=U\cup \mathcal B\cup W_M
\]
is equivalent to the original instance. Since \(M\) saturates \(W_M\), we have \(|W_M|\le |F|\le 3k\). Therefore
\[
        |Y|=|U|+|\mathcal B|+|W_M|\le 3k+k+3k=7k,
\]
so the returned instance has the desired size.  

For the running time, computing \(F\) and \(U\) can be done in \(\mathcal{O}(n^3)\) time by checking all possible triangles. Since \(k<n/6\), we have \(|F|\le 3k=\mathcal{O}(n)\). The set \(\mathcal B\) and the graph \(H\) can be constructed in \(\mathcal{O}(n|F|)=\mathcal{O}(n^2)\) time, and \(H\) has \(\mathcal{O}(n)\) vertices and \(\mathcal{O}(n^2)\) edges. A maximum matching in \(H\) can be computed in \(\mathcal{O}(n^3)\) time. Thus the kernel is computed in \(\mathcal{O}(n^3)\) time.

For the standard \MVD problem, it is easy to see that picking any edge fixes a bad triangle and so \(\mathcal B=\emptyset\). The bound then improves to \(|Y|\le 3k+3k=6k\).
\end{proof}

For the optimization version, we run the same preprocessing without using a
guessed budget. The number of edge-disjoint initially bad triangles and the
number of non-dominated outside vertices are both lower bounds on OPT. Hence
the reduced instance has at most \(6|\mathcal T|+|\mathcal B|\le 7\mathrm{OPT}\)
vertices. Moreover, Claim \ref{clm:reduction_safe} shows that any solution in the kernel can be lifted in \(\mathcal{O}(n^2) \) time to a solution of the original instance, which means that we can indeed apply our kernel before running any approximation algorithm to improve the approximation ratio. Applying the \(\mathcal{O}(\log n)\)-approximation algorithm to this reduced instance
therefore gives an \(\mathcal{O}(\log \mathrm{OPT})\)-approximation.

\section{ETH lower bound}\label{section:eth}

In this section we show that the single-exponential dependence on \(k\) in the complete-graph algorithm cannot be improved to subexponential under ETH. The reduction is for the standard version \MVD, and therefore also applies to \textsc{Interval Metric Violation} by taking sufficiently large intervals. It follows closely the NP-hardness proof of \cite{fan2018metric}. 

\ethlowerbound*

\begin{proof}
We reduce from \Problem{Vertex Cover}, for which no \(2^{o(k)}n^{\mathcal{O}(1)}\)-time algorithm exists under ETH \cite{impagliazzo2001which,cygan2015parameterized}. Let \((H=(V,E),k)\) be an instance of \Problem{Vertex Cover}. We construct a complete graph \(G\) on vertex set \(V(G)=V\cup\{r\} \), 
where \(r\) is a new vertex. For every \(v\in V\), set \(\delta(rv)=1\). For every pair \(u,v\in V\), set
\[
        \delta(uv)=
        \begin{cases}
        3, & uv\in E,\\
        2, & uv\notin E.
        \end{cases}
\]
The budget remains \(k\).

For every edge \(uv\in E\), the triangle \(ruv\) is bad, with heavy edge \(uv\), because
\[
        \delta(uv)=3>1+1=\delta(ru)+\delta(rv).
\]
For every non-edge \(uv\notin E\), the same triangle is not bad, since \(\delta(uv)=2=1+1\). Moreover, no triangle contained entirely in \(V\) is bad: all its edge weights are in \(\{2,3\}\), so the largest edge has weight at most \(3\), while the sum of the other two is at least \(4\).

Suppose first that \(H\) has a vertex cover \(X\subseteq V\) with \(|X|\le k\). Let
\[
        S=\{rv:v\in X\}.
\]
We claim that \(G-S\) has no bad cycle. Root edges have weight \(1\), and hence cannot be heavy. Edges of weight \(2\) cannot be heavy either, because every alternative path between their endpoints has weight at least \(2\). Finally, let \(uv\in E\), so \(\delta(uv)=3\). Since \(X\) is a vertex cover, at least one of \(ru,rv\) belongs to \(S\), and therefore the length-two path \(u-r-v\) is not present in \(G-S\). Every remaining \(u,v\)-path has weight at least \(3\). Thus no edge of \(G-S\) is heavier than any alternative path between its endpoints, so \(G-S\) has no bad cycle and \(S\) is a solution of size at most \(k\).
 
Conversely, suppose that the constructed instance has a feasible support \(S\) with \(|S|\le k\). We build a vertex cover \(X\) of \(H\). If \(rv\in S\), put \(v\) into \(X\). If \(uv\in S\) for an edge \(uv\in E\), put one arbitrary endpoint of \(uv\) into \(X\). Ignore selected edges corresponding to non-edges of \(H\). Then \(|X|\le |S|\le k\).

Now take any edge \(uv\in E\). The triangle \(ruv\) is bad, and \(S\) contains at least one of \(ru,rv,uv\). In the first case \(u\in X\), in the second case \(v\in X\), and in the third case we explicitly put one endpoint of \(uv\) into \(X\). Hence \(uv\) is covered by \(X\). Therefore \(X\) is a vertex cover of size at most \(k\).

The reduction is polynomial, adds only one vertex, uses weights in \(\{1,2,3\}\), and preserves the parameter. A \(2^{o(k)}n^{\mathcal{O}(1)}\)-time algorithm for \MVD would therefore imply such an algorithm for \Problem{Vertex Cover}, contradicting the standard ETH lower bound for \Problem{Vertex Cover}.
\end{proof}
\section{Ultrametric Violation Distance}\label{sec:uvd}

We record that our results for \MVD transfer to ultrametrics by a standard
blow-up of the distances.  By the classical characterization, a complete distance function is an
ultrametric if and only if in every triangle the maximum distance is attained
at least twice \cite{ardila2006bergman}. As with bad cycles in the standard metric case, this can be lifted to partial edge-length function: A partial edge-length function extends to an
ultrametric if and only if no cycle has a unique maximum-length edge. 
Thus a support
\(S\subseteq E(G)\) is feasible for \Problem{Ultrametric Violation Distance}
if and only if \(G-S\) contains no cycle whose maximum-length edge is unique.

Suppose now that \( (G, \delta,k) \) is an instance of \Problem{Ultrametric Violation Distance} and let \(\lambda_1<\cdots<\lambda_t\) be the distinct input lengths, set
\(N:=|V(G)|+1\), and define
\[
    \delta'(e):=N^i
    \qquad\text{whenever }\delta(e)=\lambda_i .
\]
Then a cycle has a unique maximum edge under \(\delta\) if and only if it is a
bad cycle under \(\delta'\): if the unique maximum has rank \(i\), then its new
length is \(N^i\), while all other edges have length at most \(N^{i-1}\), and
\(N^i>(|C|-1)N^{i-1}\); conversely, if the maximum under \(\delta\) is attained
at least twice, then no edge can be larger than the sum of the others after the
rescaling.  Hence \((G,\delta,k)\) and the standard \MVD instance
\((G,\delta',k)\) have exactly the same feasible supports.  This is the same
blow-up argument used in \cite[proof of Corollary~5.16]{cohenaddad2022fitting}.
The rescaled lengths have polynomial bit complexity.

\begin{corollary}\label{cor:uvd}
\Problem{Ultrametric Violation Distance} can be solved in
\(2^{\mathcal{O}(k)}n^{\mathcal{O}(1)}\) time, and admits a solution-lifting
kernel with at most \(6k\) vertices, computable in \(\mathcal{O}(n^3)\) time.
\end{corollary}

\begin{proof}
Apply Theorem~\ref{thm:FPT} to the support-preserving \MVD instance
\((G,\delta',k)\).  For the kernel, apply the standard-case bound in
Theorem~\ref{thm:linear_kernel} to \((G,\delta',k)\).  It returns an induced
subgraph \(G[Y]\) with \(|Y|\le 6k\).  We output \(G[Y]\) with the original
lengths \(\delta\).  The same base \(N\) still satisfies \(N>|Y|-1\), so the
support-preserving equivalence above also holds on the induced subinstance.
Thus the returned instance is a genuine ultrametric kernel, and the
solution-lifting guarantee transfers through the same equivalence.
\end{proof}

\section{Tree Metric Violation Distance}
\label{sec:tree-metric}

We now prove fixed-parameter tractability for the tree-metric variant.  The
ultrametric reduction of Section~\ref{sec:uvd} does not give a comparable
reduction here: the four-point condition forbids a strict maximum among three
sums, rather than among three individual lengths.  Thus the tree-metric case
requires a separate argument.

Throughout this section, a \emph{tree metric} means a positive tree metric.  In
other words, a distance function \(d\) on a finite set \(X\) is a tree metric if
\(d(x,y)>0\) for all distinct \(x,y\in X\), and there is a tree \(T\) with
nonnegative edge lengths and \(X\subseteq V(T)\) such that \(d(x,y)\) is the
length of the unique \(x,y\)-path in \(T\).  We do not require the points of
\(X\) to be leaves of \(T\).  If the positivity condition on distances is
dropped, we obtain the larger class of tree pseudometrics; this distinction is
relevant for the \(L_0\) objective, since the optimum tree pseudometric and tree
metric can differ by \(\Omega(n)\) modifications~\cite[Sec.~8.2]{cohenaddad2024fitting}.

\begin{problembox}{Tree Metric Violation Distance}
\textbf{Input:} A complete graph \(G\), an edge-length function
\(\delta:E(G)\to \mathbb R_{>0}\), and an integer \(k\).\\
\textbf{Parameter:} \(k\).\\
\textbf{Question:} Is there a set \(S\subseteq E(G)\), \(|S|\le k\), such that
the lengths of the edges in \(S\) can be changed so that the resulting distance
function is a tree metric?
\end{problembox}

As before, we call \(S\subseteq E(G)\) a \emph{feasible support} if the partial
distance function \(\delta|_{E(G)\setminus S}\) has a completion to a tree
metric.  Feasible supports are closed under taking supersets.  The goal is to
decide whether a feasible support of size at most \(k\) exists.

\treefpttheorem*

The proof uses a branching algorithm whose nodes are partial supports.  The key
is the following oracle.  Given a partial support \(S\), the oracle either
constructs a tree-metric completion, or returns a bounded set of
pairs, one of which must be added by every feasible extension of \(S\).

Let \(Z\) be the set of vertices incident with an edge of \(S\), and let
\(W:=V(G)\setminus Z\).  Then \(|Z|\le 2k\), and every pair meeting \(W\) is
fixed.  If \(\delta|_W\) is a tree metric, the vertices of \(W\) determine a
unique reduced tree \(T_W\).  Each vertex \(z\in Z\) is then forced, by its
fixed distances to \(W\), to attach to a point \(p_z\in T_W\) with an edge of length \(\lambda_z\).  The only remaining freedom is local:
vertices with the same attachment point may still be arranged among themselves.
The lemmas below make this statement algorithmic and turn failures into small
branching certificates.

\paragraph*{Tree-metric preliminaries.}
We use the following standard facts.  A finite positive metric is a tree metric
if and only if it satisfies the four-point condition: for all distinct
\(a,b,c,e\), the maximum of
\[
    d(ab)+d(ce),\qquad d(ac)+d(be),\qquad d(ae)+d(bc)
\]
is attained at least twice~\cite{buneman1971recovery,semple2003phylogenetics}.
Moreover, every finite tree metric has a unique \emph{reduced} representing
tree, obtained by requiring all edges to have positive length, every leaf to
belong to the ground set, and no degree-two vertex outside the ground set.  This
reduced tree can be constructed in polynomial time by classical exact
reconstruction algorithms~\cite{waterman1977additive}.  Finally, if a positive
distance function on a set \(Y\) is not a tree metric, then a violating triple
or quadruple can be found by inspecting all triples and quadruples of \(Y\).

\begin{observation}\label{obs:tree-restriction}
Let \(d\) be a tree metric on \(X\) with representing tree \(T\), and let
\(Y\subseteq X\), \(|Y|\ge 2\).  Then \(d|_Y\) is a tree metric.  The subtree of
\(T\) spanned by \(Y\) is isometrically embedded in \(T\), and the reduced
representing tree of \(d|_Y\) is obtained from it by pruning leaves outside
\(Y\), suppressing degree-two vertices outside \(Y\), and contracting zero-length
edges.
\end{observation}

We will also use the following elementary observation:

\begin{observation}\label{obs:tree-local-witness}
Fix \(S\subseteq E(G)\), and let \(X_0\subseteq V(G)\).  Suppose that no tree
metric on \(X_0\) agrees with \(\delta\) on
\(E(G[X_0])\setminus S\).  Then
\(F:=E(G[X_0])\setminus S\) intersects every feasible support \(S'\supseteq S\).
\end{observation}

\paragraph*{Anchors.}
Assume that \(|W|\ge 2\) and that \(\delta|_W\) is a tree metric, and let
\(T_W\) be its reduced representing tree.  For \(A\subseteq W\), \(|A|\ge 2\),
write \(\Sigma(A)\) for the subtree of \(T_W\) spanned by \(A\).  Every point of
\(\Sigma(A)\) lies on a path between two vertices of \(A\). The next lemma says that, once the attachment point is known to lie in
\(\Sigma(A)\), the distances to \(A\) determine both the attachment point and
the pendant length.
\begin{lemma}[Anchor rigidity]\label{clm:anchor}
Let \(A\subseteq W\) with \(|A|\ge 2\), and let
\(p,q\in\Sigma(A)\) and \(\lambda,\lambda'\ge 0\) satisfy
\[
        \lambda+\dist_{T_W}(q,a)=\lambda'+\dist_{T_W}(p,a)
        \qquad\text{for every } a\in A .
\]
Then \(q=p\) and \(\lambda=\lambda'\).
\end{lemma}

\begin{proof}
Choose \(a_1,a_2\in A\) such that \(q\) lies on the \(a_1,a_2\)-path.  Adding
the two corresponding equalities gives
\[
    2\lambda+\dist_{T_W}(a_1,a_2)
    =2\lambda'+\dist_{T_W}(p,a_1)+\dist_{T_W}(p,a_2)
    \ge 2\lambda'+\dist_{T_W}(a_1,a_2),
\]
and hence \(\lambda\ge\lambda'\).  By symmetry, using two vertices of \(A\) whose
path contains \(p\), we also get \(\lambda'\ge\lambda\).  Thus
\(\lambda=\lambda'\), and the displayed hypothesis gives
\(\dist_{T_W}(q,a)=\dist_{T_W}(p,a)\) for all \(a\in A\).  In particular, equality
holds in the triangle inequality above, so \(p\) also lies on the
\(a_1,a_2\)-path.  Two points on this path with the same distance from \(a_1\)
coincide, so \(p=q\).
\end{proof}

The condition \(p,q\in\Sigma(A)\) is essential.  If \(p\) lies outside
\(\Sigma(A)\) and \(q\) is its projection to \(\Sigma(A)\), then distances from
\(p\) to all vertices of \(A\) differ from the corresponding distances from
\(q\) by the same additive constant.  Thus an attachment at \(p\) is
indistinguishable from an attachment at \(q\) after changing the pendant length
by that constant.

For every point \(p\) of \(T_W\), fix two vertices \(A_p:=\{a,b\}\subseteq W \)
    such that \(p \) lies on the unique  \(a,b\) path.
Such vertices exist because \(|W|\ge 2\), every point of a tree lies on a path
between two leaves, and every leaf of the reduced tree \(T_W\) belongs to \(W\).
Then \(p\in\Sigma(A)\) for every \(A\supseteq A_p\).

The following lemma is a core subroutine of our algorithm. It shows how to check if a set $S$ is a solution on a small set of vertices: 
\begin{lemma}\label{lem:tree-small-completion}The following two problems can be
solved in \(2^{\mathcal{O}(m\log m)}m^{\mathcal{O}(1)}\) time.

\begin{enumerate}
    \item \emph{Partial tree-metric completion.}
    Given a set \(X_0\) of \(m\) terminals, a set
    \(F_0\subseteq \binom{X_0}{2}\) of fixed pairs, and positive values
    \(\mu(xy)>0\) for all \(xy\in F_0\), decide whether there is a tree metric
    \(d\) on \(X_0\) such that
    \[
        d(x,y)=\mu(xy) \qquad\text{for every }xy\in F_0 .
    \]

    \item \emph{Rooted partial tree completion.}
    Given a set \(R\) of \(m\) terminals, a root symbol \(r\notin R\),
    prescribed values \(h(z)\ge 0\) for all \(z\in R\), a set
    \(F_0\subseteq \binom{R}{2}\) of fixed terminal pairs, and positive values
    \(\mu(zz')>0\) for all \(zz'\in F_0\), decide whether there is a tree
    containing the terminals \(R\) and a marked root \(r\) such that all edges
    have positive length,
    \[
        \dist(r,z)=h(z) \qquad\text{for all }z\in R,
    \]
    \[
        \dist(z,z')=\mu(zz') \qquad\text{for all }zz'\in F_0,
    \]
    and the terminals in \(R\) are represented by distinct vertices.
    The root \(r\) is allowed to be placed on the same vertex as one terminal
    of \(R\).
\end{enumerate}

In both cases, if the answer is yes, a witnessing tree can also be returned
within the same running time.
\end{lemma}

\begin{proof}
We first consider the partial tree-metric completion problem. We enumerate all reduced marked tree topologies for the terminal set \(X_0\).
By this we mean tree topologies in which the labels of \(X_0\) are placed
injectively on vertices, every leaf carries a label, and every unlabelled
vertex has degree at least three. Labels are not required to be leaves: a
terminal may be an internal vertex of the representing tree.

There are only \(2^{\mathcal{O}(m\log m)}\) such topologies. Indeed, every
such topology has at most \(2m-2\) vertices: if there are \(m\) labelled
vertices, every leaf is labelled, and every unlabelled vertex has degree at
least three, then the usual degree-counting argument bounds the number of
unlabelled vertices by \(m-2\). Thus we can enumerate, for each
\(t\le 2m-2\), all labelled trees on \(t\) auxiliary vertices and all placements
of the \(m\) terminal labels, and keep only the reduced marked tree topologies.
The number of possibilities is \(2^{\mathcal{O}(m\log m)}\).

Fix one enumerated topology \(\tau\). We introduce a variable \(\ell_e\) for
every edge \(e\) of \(\tau\). For \(x,y\in X_0\), let
\[
    D^\tau_{xy}(\ell)=\sum_{e\in \tau[x,y]}\ell_e
\]
be the length of the unique \(x\)-\(y\) path in \(\tau\). We impose the linear
equations
\[
    D^\tau_{xy}(\ell)=\mu(xy)
    \qquad\text{for every }xy\in F_0 .
\]

It remains to enforce that the chosen topology has positive edge lengths.
To avoid strict inequalities, introduce one additional variable \(\eta\), and
solve the linear program
\[
\begin{array}{rl}
\text{maximize} & \eta \\[2mm]
\text{subject to} &
D^\tau_{xy}(\ell)=\mu(xy)
    \qquad\text{for every }xy\in F_0,\\
& \ell_e\ge \eta
    \qquad\text{for every edge }e\in E(\tau),\\
& 0\le \eta\le 1 .
\end{array}
\]
The topology \(\tau\) is feasible if and only if the optimum value is strictly
positive. If the optimum is positive, all edges of \(\tau\) have positive
length, and therefore all distances between distinct terminals of \(X_0\) are
positive. Conversely, any tree-metric completion whose reduced representing
tree has topology \(\tau\) has all edge lengths positive, so the minimum edge
length gives a feasible positive value of \(\eta\), after truncating at \(1\)
if necessary.

Thus feasibility for a fixed topology can be checked by linear programming.
If the test succeeds, the edge lengths returned by the linear program give a
witnessing tree. Since the number of topologies is
\(2^{\mathcal{O}(m\log m)}\), the claimed running time follows for the first
problem.

The rooted partial completion problem is similar. We enumerate reduced marked
tree topologies with terminal labels \(R\) and one additional root label \(r\).
The labels in \(R\) must be placed injectively on vertices, but the root label
\(r\) may be placed either on an unlabelled vertex or on a vertex already
carrying one terminal label. Again every leaf must carry at least one label,
and every unlabelled vertex must have degree at least three. The same counting
argument gives \(2^{\mathcal{O}(m\log m)}\) topologies.

Fix such a rooted topology \(\tau\). For labels \(a,b\in R\cup\{r\}\), let
\(D^\tau_{ab}(\ell)\) denote the length of the path between their marked
vertices. If \(r\) is placed on the same vertex as some terminal \(z\), then
\(D^\tau_{rz}(\ell)=0\).

We impose the equations
\[
    D^\tau_{rz}(\ell)=h(z)
    \qquad\text{for every }z\in R,
\]
and
\[
    D^\tau_{zz'}(\ell)=\mu(zz')
    \qquad\text{for every }zz'\in F_0 .
\]
As before, we introduce \(\eta\) and maximize it subject to these equations,
\[
    \ell_e\ge \eta
    \qquad\text{for every edge }e\in E(\tau),
\]
and
\[
    0\le \eta\le 1 .
\]
The topology is feasible if and only if the optimum value is strictly positive.

The distinction from the first problem is that the root is not required to be a
separate terminal of a positive metric. Thus \(h(z)=0\) is allowed: it is
represented by placing the root label \(r\) on the same vertex as \(z\).
However, two distinct terminals of \(R\) are never placed on the same vertex,
and all edges have positive length, so the distance between any two distinct
terminals in \(R\) is automatically positive.

Finally, any rooted completion with nonnegative edge lengths can be converted
to one of the topologies enumerated above by contracting all zero-length edges
and then suppressing unlabelled degree-two vertices. The terminals in \(R\)
remain distinct because their pairwise distances are required to be positive;
only the root may become identified with a terminal. This proves correctness
of the rooted test and completes the proof.
\end{proof}

We are now ready to prove the main lemma of the search tree. 

\begin{lemma}
\label{lem:tree-partial-completion}
Let \(S\subseteq E(G)\), \(|S|\le k\).  There is an algorithm running in
\(k^{\mathcal{O}(k)}n^{\mathcal{O}(1)}\) time that decides whether
\(\delta|_{E(G)\setminus S}\) has a tree-metric completion.  If no completion
exists, the algorithm outputs a set \(F\subseteq E(G)\setminus S\),
\(|F|=\mathcal{O}(k^2)\), that intersects every feasible support
\(S'\supseteq S\).
\end{lemma}

\begin{proof}
Let \(Z\) be the set of vertices incident with an edge of \(S\), and let
\(W:=V(G)\setminus Z\).  Thus \(|Z|\le 2k\), and every pair meeting \(W\) is
fixed.

If \(|W|\le 1\), then \(|V(G)|\le 2k+1\).  We apply Lemma~\ref{lem:tree-small-completion}
to \(X_0:=V(G)\), \(F_0:=E(G) \setminus S\), and the prescribed values \(\mu=\delta\) on \(E(G)\setminus S\).  If it succeeds, we have a completion.  If it
fails, then \(F:=E(G)\setminus S\) has size \(\mathcal{O}(k^2)\) and is valid by
Observation~\ref{obs:tree-local-witness}.  We may therefore assume from now on
that \(|W|\ge 2\).

First test whether \(\delta|_W\) is a tree metric.  If not, let
\(X_0\subseteq W\) be a violating triple or quadruple.  Since all edges of
\(G[X_0]\) are fixed, Observation~\ref{obs:tree-local-witness} allows us to
return \(F:=E(G[X_0])\), which has size at most six.  Otherwise let \(T_W\) be
the reduced representing tree of \(\delta|_W\).

Next consider each \(z\in Z\).  Test whether \(\delta|_{W\cup\{z\}}\) is a tree
metric.  If this fails, choose a violating triple or quadruple
\(X_0\subseteq W\cup\{z\}\).  Since \(\delta|_W\) is already a tree metric, the
set \(X_0\) contains \(z\); every edge of \(G[X_0]\) is fixed, because the only
vertex of \(Z\) in \(X_0\) is \(z\).  Again we return
\(F:=E(G[X_0])\).
 
Let \(T_{W,z}\) be the reduced representing tree of
\(\delta|_{W\cup\{z\}}\). The minimal subtree of \(T_{W,z}\) spanning \(W\),
after suppressing unlabelled degree-two vertices, is canonically isometric to
\(T_W\). Under this identification, let \(p_z\) be the first point of \(T_W\)
met by the path from \(z\) to \(T_W\), and let
\(\lambda_z=\dist_{T_{W,z}}(z,p_z)\). Then, for every \(w\in W\): 
\begin{equation}
    \label{eq:attachment}
\delta(zw)=\lambda_z+\dist_{T_W}(p_z,w)
\end{equation}

Uniqueness follows from Lemma~\ref{clm:anchor}, applied with \(A=W\).
Both \(p_z\) and \(\lambda_z\) can be read from \(T_{W,z}\), and hence are
computable in polynomial time. Moreover, every completion extending \(S\) restricts to the same fixed tree
metric on \(W\cup\{z\}\), and therefore induces the same pair
\((p_z,\lambda_z)\).

Now consider every pair \(z,z'\in Z\) with distinct gates
\(p_z\neq p_{z'}\).  In any completion, the path from \(z\) to \(z'\) must pass
through \(p_z\) and \(p_{z'}\). Indeed, if the branches containing \(z\) and \(z'\) met before reaching their
distinct gates, their union with the \(p_z,p_{z'}\)-path in \(T_W\) would
contain a cycle. Therefore every completion forces
\begin{equation}
\label{eq:forced-distinct-projections}
    d(zz')=
    \lambda_z+\dist_{T_W}(p_z,p_{z'})+\lambda_{z'} .
\end{equation}
If \(zz'\in S\), no test is needed.  If \(zz'\notin S\) and the fixed value
\(\delta(zz')\) differs from the right-hand side of
\eqref{eq:forced-distinct-projections}, we build a constant-size certificate.
Let
\[
    A:=A_{p_z}\cup A_{p_{z'}},
    \qquad X_0:=\{z,z'\}\cup A .
\]
Then \(|X_0|\le 6\), and every edge of \(G[X_0]\) is fixed. Moreover, since \(A\) contains both \(A_{p_z}\) and \(A_{p_{z'}}\), we have
\(p_z,p_{z'}\in\Sigma(A)\). We claim that no
tree metric on \(X_0\) agrees with \(\delta\).  Indeed, in any such tree metric,
the subtree spanned by \(A\) identifies isometrically with \(\Sigma(A)\subseteq
T_W\).  Let \(q\in\Sigma(A)\) be the point where the path from \(z\) first meets
this subtree, and let \(\lambda=d(z,q)\).  For every \(a\in A\),
\[
    \delta(za)=\lambda+\dist(q,a)
              =\lambda_z+\dist_{T_W}(p_z,a).
\]
 Since
\(p_z,q\in\Sigma(A)\), Lemma~\ref{clm:anchor} gives \(q=p_z\) and
\(\lambda=\lambda_z\).  The same argument applies to \(z'\).  Thus the distance
between \(z\) and \(z'\) is forced to be the right-hand side of
\eqref{eq:forced-distinct-projections}, contradicting the fixed value
\(\delta(zz')\).  By Observation~\ref{obs:tree-local-witness}, we return
\(F:=E(G[X_0])\setminus S\).

After all distinct-gate distances have been verified, constraints involving
different gate classes are completely determined by \(T_W\) and the values
\(\lambda_z\). Consequently, the remaining completion problems for the
different sets \(Z_p\) are independent. It remains to handle vertices with a common gate. For each attachment point
\(p\), let \(Z_p:=\{z\in Z:p_z=p\}\).  The only constraints not yet accounted for
inside this group are the distances among vertices of \(Z_p\).  They form the
following rooted completion instance on \(Z_p\cup\{r_p\}\), where \(r_p\) is a
formal root representing \(p\):
\[
    d(r_p,z)=\lambda_z \quad (z\in Z_p),
    \qquad
    d(zz')=\delta(zz') \quad \text{for } zz'\notin S,
\]
and pairs \(zz'\in S\cap E(G[Z_p])\) are free.  The vertices of \(Z_p\) must be
pairwise distinct, while \(r_p\) is allowed to coincide with a vertex \(z\) when
\(\lambda_z=0\).  If \(|Z_p|\le 1\), this rooted instance is feasible.  Otherwise
we solve it using the rooted variant of Lemma~\ref{lem:tree-small-completion},
in time \(k^{\mathcal{O}(k)}\).
 
Suppose the rooted instance for some \(p\) is infeasible.  Set
\(X_0:=Z_p\cup A_p\).  We claim that no tree metric on \(X_0\) agrees with
\(\delta\) on \(E(G[X_0])\setminus S\).  As above, Lemma~\ref{clm:anchor}, applied
with \(A=A_p\), forces every \(z\in Z_p\) to attach at \(p\) with pendant length
\(\lambda_z\).  Restricting such a tree to \(Z_p\cup\{p\}\) would give a feasible
rooted completion, a contradiction.  Therefore
\(F:=E(G[X_0])\setminus S\) is valid by Observation~\ref{obs:tree-local-witness};
its size is \(\mathcal{O}(k^2)\).

Suppose now that all the previous tests were successful and let us show how to construct the final tree from \(T_W\).  Start with \(T_W\).  For each
attachment point \(p\), take a feasible rooted tree for \(Z_p\cup\{r_p\}\) and
identify \(r_p\) with the point \(p\) of \(T_W\), subdividing edges of \(T_W\) as
needed when attachment points lie in their interiors.  The resulting weighted
tree realises \(\delta\) on \(W\) by construction, realises every fixed distance
from \(Z\) to \(W\) by \eqref{eq:attachment}, realises every fixed distance
between distinct gate classes by \eqref{eq:forced-distinct-projections}, and
realises the fixed distances inside each class \(Z_p\) by the rooted completion.
Thus it is a tree-metric completion of \(\delta|_{E(G)\setminus S}\).

The only super-polynomial work is the enumeration of tree topologies on
\(\mathcal{O}(k)\) points in Lemma~\ref{lem:tree-small-completion}.  Hence the
running time is \(k^{\mathcal{O}(k)}n^{\mathcal{O}(1)}\), and every returned
certificate has size \(\mathcal{O}(k^2)\).
\end{proof}

\begin{proof}[Proof of Theorem~\ref{thm:tree-metric-fpt}]
We build a bounded search tree.  A node is a set \(S\subseteq E(G)\) of pairs
chosen so far.  The root is \(S=\emptyset\).  If \(|S|>k\), we reject the node.
Otherwise we apply Lemma~\ref{lem:tree-partial-completion}.  If a completion
exists, then \(S\) is a feasible support and we accept.  If no completion
exists, the lemma returns a set \(F\subseteq E(G)\setminus S\) of size
\(\mathcal{O}(k^2)\) that intersects every feasible support extending \(S\); we
branch on all choices \(e\in F\) and recurse on \(S\cup\{e\}\).

For correctness, let \(S^*\) be any feasible support of size at most \(k\).  At
any node with \(S\subseteq S^*\), either the oracle finds a completion and the
algorithm accepts, or it returns a set \(F\) meeting \(S^*\setminus S\).  Hence
one child again satisfies \(S\cup\{e\}\subseteq S^*\).  Therefore the search tree
reaches \(S^*\), unless it accepts earlier.

The depth is at most \(k\), and the branching factor is \(\mathcal{O}(k^2)\).
Each node is processed in \(k^{\mathcal{O}(k)}n^{\mathcal{O}(1)}\) time by
Lemma~\ref{lem:tree-partial-completion}.  This gives the claimed running time.
\end{proof}

\section{Concluding Remarks}

The main open problem remains whether \MVD admits a constant-factor
approximation in polynomial time~\cite{fan2018metric,cohenaddad2022fitting}.
Our kernel shows that it is enough to solve this question on instances with
\(\mathcal{O}(\mathrm{OPT})\) vertices: running the kernel at
\(k=\mathrm{OPT}\) preserves the optimum and leaves at most
\(6\,\mathrm{OPT}\) vertices.  Thus the obstacle is not large ground sets, but
rather dense instances in which the optimum is already linear in the number of
vertices.

Several quantitative questions also remain open.  We did not try to optimize
the constants in the kernels: can the bounds \(6k\) and \(7k\) be improved?
Similarly, the randomized FPT algorithm has success probability
\(2^{-11k}\), coming from the random-separation bound
\(|O\cup Q|\le 11k\).  Although ETH rules out a subexponential dependence on
\(k\), the best possible base of the exponent is unknown.

For structured metrics, the parameterized picture is still incomplete.
Ultrametrics follow from \MVD by the rescaling of Section~\ref{sec:uvd}, while
tree metrics require a separate argument and currently yield only a
\(k^{\mathcal{O}(k)}n^{\mathcal{O}(1)}\)-time algorithm.  The two bottlenecks
are the \(\mathcal{O}(k^2)\)-size obstruction in
Lemma~\ref{lem:tree-partial-completion} and the enumeration of
\(k^{\mathcal{O}(k)}\) tree topologies.  Whether
\Problem{Tree Metric Violation Distance} admits a
\(2^{\mathcal{O}(k)}n^{\mathcal{O}(1)}\)-time algorithm, or a polynomial
kernel, remains open.

\section*{Declarations}
\paragraph{Use of generative AI.}
We used Claude (Anthropic) and OpenAI ChatGPT  to assist with
(i)~conducting a literature survey,
(ii)~tightening and improving the language and exposition throughout the paper,
(iii)~suggesting formulations for parts of the introduction, 
(iv)~generating drafts of proofs throughout the paper; in all cases based on scaffolding, proof strategies, and technical direction provided by the authors, and
(vi)~preparing an initial draft for this disclosure statement. 
The tool materially affected the exposition across all sections of the paper.
The authors verified the correctness and originality of all content, including all mathematical statements, proofs, and references.

\bibliographystyle{alpha}
\bibliography{bib_metric}

\end{document}